\documentclass[10pt,twocolumn,twoside]{IEEEtran}
\makeatletter
\def\ps@headings{%
\def\@oddhead{\mbox{}\scriptsize\rightmark \hfil \thepage}%
\def\@evenhead{\scriptsize\thepage \hfil \leftmark\mbox{}}%
\def\@oddfoot{}%
\def\@evenfoot{}}
\makeatother
\usepackage{paralist}

\usepackage[ruled,vlined]{algorithm2e}

\usepackage{amssymb}
\usepackage{amsmath} 

\usepackage{color}
\usepackage{graphicx}
\usepackage{epstopdf}
\usepackage{stmaryrd}
\usepackage{dsfont}
\usepackage{subfig}
\usepackage{paralist}

\usepackage{soul}

\newif\iflong 
\longfalse

\newif\ifcomm
\commtrue

\newtheorem{lem}{Lemma}

\newcommand{\set}[1]{\left\{#1\right\}}

\newenvironment{proof-sketch}{{\noindent\em Proof Sketch.}\hspace*{0.3em}}{\qed\medskip}
\newenvironment{proof}{{\bf Proof.}}{\hfill\rule{2mm}{2mm}\\}
\newenvironment{proof-of}[1]{{\noindent\em Proof of #1.}\hspace*{0.3em}}{\qed\medskip}

\newcounter{assumption}
\newcommand{\theassumptionletter}{A}
\renewcommand{\theassumption}{\theassumptionletter\arabic{assumption}}

\newcommand{\MF}{\mathcal{F}}

\newcommand{\MP}{\mathcal{P}}

\newcommand{\beq}{\begin{equation}}
\newcommand{\eeq}{\end{equation}}
\newcommand{\beqa}{\begin{eqnarray}}
\newcommand{\eeqa}{\end{eqnarray}}
\newcommand{\beqan}{\begin{eqnarray*}}
\newcommand{\eeqan}{\end{eqnarray*}}
\newcommand{\ben}{\begin{eqnarray*}}
\newcommand{\een}{\end{eqnarray*}}

\ifcomm
   \newcommand\comm[1]{\textcolor{blue}{ #1}}
\else
   \newcommand\comm[1]{}
   \renewcommand{\todo}[1]{}
\fi

\newcommand{\bm}{\boldmath}

\newcommand{\qed}{{\hfill$\Box$}}

\newcommand{\xx}{\mbox{$\mathbf{x}$}}

\newcommand{\nn}{\mbox{$\mathbf{n}$}}
\newcommand{\zz}{\mbox{$\mathbf{z}$}}

\newcommand{\vv}{\mbox{\bm $v$}}

\newcommand{\yy}{\mbox{$\mathbf y$}}
\newcommand{\BH}{\mbox{\bm $H$}}
\newcommand{\C}{\mbox{\bm $C$}}
\newcommand{\G}{\mbox{\bm $G$}}

\newcommand{\A}{\mbox{\bm $A$}}
\newcommand{\B}{\mbox{\bm $B$}}

\newcommand{\X}{\mbox{\bm $X$}}
\newcommand{\Y}{\mbox{\bm $Y$}}

\newcommand{\bica}{{bICA}}

\makeatletter
\renewcommand\paragraph{\@startsection{paragraph}{4}{\z@}%
    {1.5ex plus .2ex minus .3ex}%
            {-0em}%
                        {\normalsize\bf}}

\newcommand{\captionfonts}{\small}

\long\def\@makecaption#1#2{%
  \vskip\abovecaptionskip
  \sbox\@tempboxa{{\captionfonts #1: #2}}%
  \ifdim \wd\@tempboxa >\hsize
    {\captionfonts #1: #2\par}
  \else
    \hbox to\hsize{\hfil\box\@tempboxa\hfil}%
  \fi
  \vskip\belowcaptionskip}
\makeatother

\title{\huge{Binary Inference for Primary User Separation in Cognitive Radio Networks}
\thanks{An earlier version of this work appeared in the Proceedings of the $5^{th}$ International Conference on Cognitive Radio Oriented Wireless Networks and Communications (Crowncom 2010).}
}
\author{
\begin{tabular}{c}
\\
Huy Nguyen$^{\dagger}$, Guanbo Zheng$^{\ddagger}$, Zhu Han$^{\ddagger}$, and Rong Zheng$^{\dagger}$ \\
$^{\dagger}$Department of Computer Science \\
$^{\ddagger}$Department of Electrical and Computer Engineering \\
University of Houston, Houston, TX 77204\\
E-mail: {\it \{hanguyen5, gzheng3, zhan2, rzheng\}@uh.edu}
\end{tabular}
}
\begin{document}
\maketitle
\begin{abstract}
Spectrum sensing receives much attention recently in the cognitive radio (CR) network
research, i.e., secondary users (SUs) constantly monitor channel condition to
detect the presence of the primary users (PUs). In this paper, we go beyond
spectrum sensing and introduce the \textit{PU separation problem}, which
concerns with the issues of distinguishing and characterizing PUs in the
context of collaborative spectrum sensing and monitor selection.  The
observations of monitors are modeled as boolean OR mixtures of underlying
binary sources for PUs.  We first justify the use of the binary OR mixture
model as opposed to the traditional linear mixture model through simulation
studies. Then we devise a novel binary inference algorithm for PU separation.
Not only PU-SU relationship are revealed, but PUs' transmission statistics and
activities at each time slot can also be inferred.  Simulation results show
that without any prior knowledge regarding PUs' activities, the algorithm
achieves high inference accuracy even in the presence of noisy measurements.
\end{abstract}

\section{Introduction}
%

With tremendous growth in wireless services, the demand for radio spectrum
has significantly increased. However, spectrum resources are scarce and
most of them have been already licensed to existing operators. Recent studies
have shown that despite claims of spectral scarcity, the actual licensed
spectrum remains unoccupied for long periods of time~\cite{FCC}. Thus,
cognitive radio (CR) systems have been proposed~\cite{CR00,CR01,CR02} in order to
efficiently exploit these spectral holes{, in which licensed primary users (PUs) are not present}. CRs or secondary users (SUs) are
wireless devices that can intelligently monitor and adapt to their environment,
hence, they are able to share the spectrum with the licensed PUs, operating when the PUs are idle.

One key challenge in CR systems is spectrum sensing, i.e., SUs attempt to
learn the environment and determine the presence and characteristics of PUs.
Spectrum sensing can be done at SUs individually or
cooperatively \cite{Ghasemi_Sousa,Sun_Zhang}, with or without the assistance of infrastructure supports
such as dedicated monitor nodes and cognitive pilot channel (CPC) \cite{CPC1,CPC2,CPC3,CPC4}.
Energy detection is one of the most commonly used method for spectrum sensing, where
the detector computes the energy of the received signals and
compares it to a certain threshold value to decide whether the PU signal
is present or not. It has the advantage of short detection time
but suffers from low accuracy compared to feature-based approaches such as cyclostationary
detection~\cite{CR01,CR02}.  From the prospective of a CR system, it is often
insufficient to detect PU activities in a single SU's vicinity (``is there any
PU near me?").  Rather, it is important to determine the identity of PUs (``who
is there?") as well as the distribution of PUs in the field (``where are
they?"). We call these issues the {\it PU separation problem}.

To motive the need for PU separation, let us consider the following scenarios:
\begin{itemize}
\item Multiple SUs cooperatively infer the activities of PUs, some of which may
be observable to only a subset of SUs. In this case, the SUs need to identify
the PU-SU adjacency relationships. Blindly assuming all PUs are observable to
all SUs will lead to inferior detection results.
\item Dedicated monitors are employed for spectrum sensing. There exists
redundancy in monitors' observations due to common PUs across multiple
monitors. Such redundancy can be reduced by judiciously selecting a subset of
monitors to report their spectrum sensing results. Furthermore, some monitors
can be put to low-power modes for energy conservation.
\end{itemize}

Clearly, PU separation is a more challenging problem compared to node-level
PU detection. The conventional wisdom suggests that sophisticated methods such as
feature-based detection are necessary. On the contrary, we find that through
cooperation among monitors or SUs, not only accuracy of energy detection can
be improved as been demonstrated in several existing work \cite{Wang10, Ghasemi_Sousa, Sun_Zhang}, but also PUs
can be identified using solely binary information (due to thresholding in
energy detection). The key to this surprising result is a binary
inference framework that models the observations of SUs and monitors as boolean
OR mixtures of underlying binary latency sources for PUs. It allows us to
exploit the correlation structure among distributed binary observations. We
develop an iterative algorithm, called Binary Independent Component Analysis
(\bica), to determine the underlying latent sources (i.e., PUs) and their active
probabilities. In \bica, no prior information regarding the PUs' activities or
observation noise is assumed. Given $m$ monitors or SUs, up to $2^m-1$ PUs can
be inferred resulting in great efficiency.  
Evaluation results show effectiveness of \bica~under practical settings.

\paragraph*{Contributions} In this paper, we make the following contributions
toward the design of a binary inference framework for PU separation in cognitive
radio networks:

\begin{itemize}
\item We introduce the PU separation problem with cooperative SU inference model
and discuss its importance on CR systems.
\item We provide a stochastic analysis on the difference between linear and binary
PU energy detection models. Results from the study imply that using just binary observations
from SUs has comparable accuracy with using a linear model, while incurring
much less overhead.
\item We apply \bica~to solve the PU separation problem
without any assumption on the noise model or prior
knowledge on the PU activities. We furthermore
consider the inverse problem of inferring the detailed PUs' activities given the
SUs' observations and the inferred model.
\end{itemize}

The rest of the paper is organized as follows. In Section~\ref{sect:model}, the
observation model is introduced. A comparison between the linear and binary energy
model and brief overview of related work are also presented.
In Section~\ref{sect:algo}, we present the
\bica~algorithm to determine the statistics of PU activities and the inference
algorithm to decide which set of PUs are active.
Formulation and solution to the inverse problem under noisy
measurements are presented in Section~\ref{sec:inverse}.
Evaluation results are
detailed in Section~\ref{sect:eval} followed by an overview of related work is
provided in Section~\ref{sec:related} followed by conclusions in
Section~\ref{sect:conclusion}.

\section{Model and Preliminary}
\label{sect:model}
Consider a slotted system in which the transmission activities of $n$ PUs are
modeled as a set of independent binary variables $\yy$ with active
probabilities $\MP(\yy)$. The binary observations due to energy detection at
the $m$ monitor nodes (for the remaining of the paper, we do not distinguish
monitor nodes and SUs) are modeled as an $m$-dimension binary vector $\xx = [x_1,x_2,\ldots,x_m]^T$ with joint
distribution $\MP(\xx)$. It is assumed that presence of any active PU
surrounding of a monitor leads to positive detection.  An unknown binary
mixing matrix $\G_{m\times n}$ is used to represent the relationship between the
observable variables in $\xx$ and the latent binary variables in $\yy =
[y_1,y_2,\ldots,y_n]^T$ as follows:
\beq
x_i = \bigvee_{j=1}^{n}{(g_{ij}\wedge y_j)}, \mbox{ $i = 1, \ldots, m$},
\label{eq:boolean}
\eeq
where $\wedge$ is Boolean {\em AND}, $\vee$ is Boolean {\em OR}, and $g_{ij}$
is the entry on the $i$th row and the $j$th column of $\G$.  For ease of
presentation, we introduce a short-hand notation as
\beq
\xx = \G \otimes \yy.
\eeq
In essence, $g_{ij}$ encodes whether monitor $i$ can detect the transmission
of PU $j$. For a monitor $i$, the energy detection returns 1 when {the monitor can detect one or more active PUs}.
$\G$ can be {seen} as the adjacency matrix of an undirected
bi-partite graph $G=(U,V,E)$, where $U = \set{x_1, x_2, \ldots, x_m}$ and $V =
\set{y_1, y_2, \ldots, y_n}$. An edge $e=(x_i,y_j)$ exists  if $g_{ij} = 1$.
Illustration of a sample network scenario and its bipartite graph is presented in Figure~\ref{fig:model}.

\begin{figure}[tp]
\begin{center}
\includegraphics[width=2.4in]{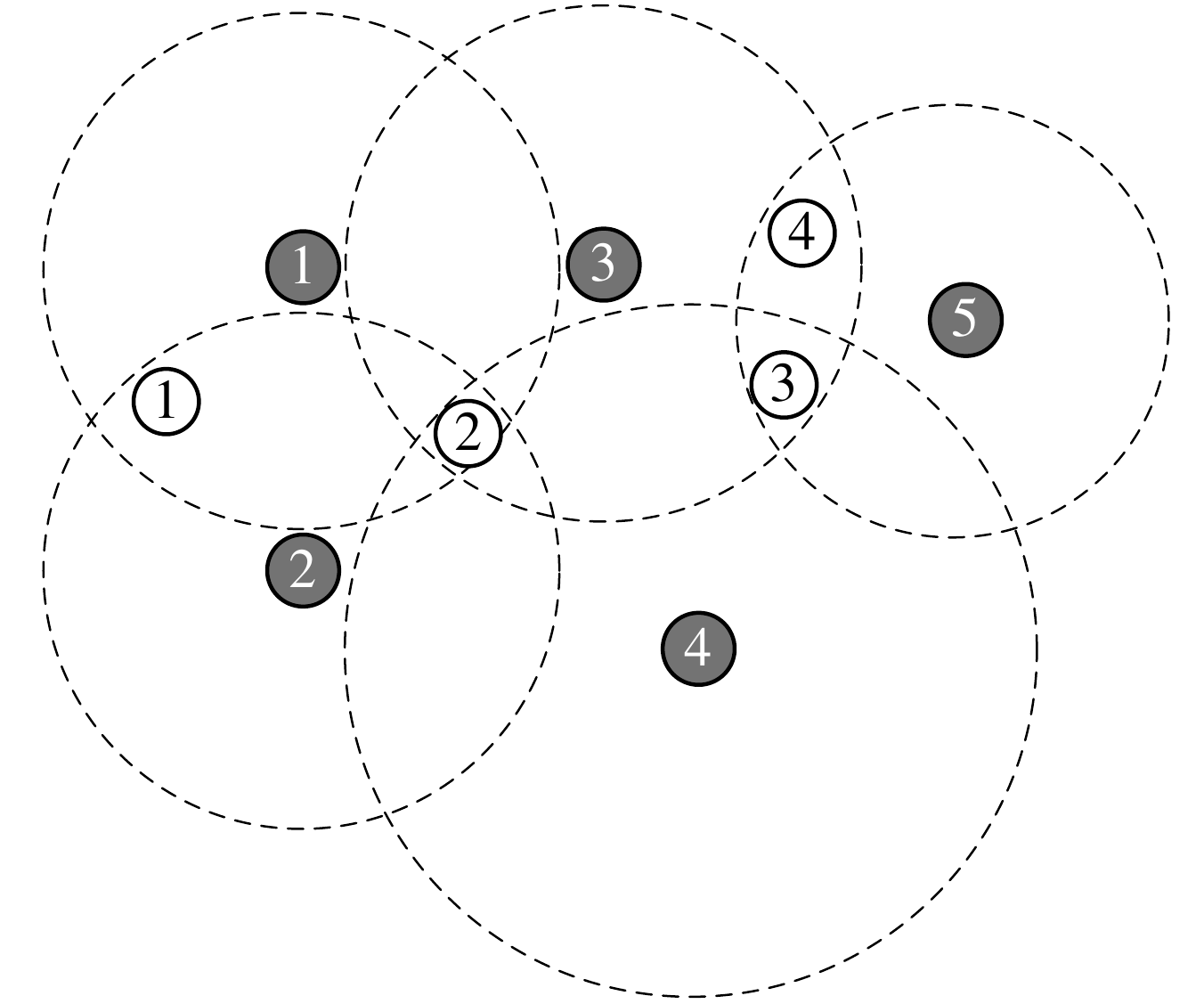}\\
\includegraphics[width=2.8in]{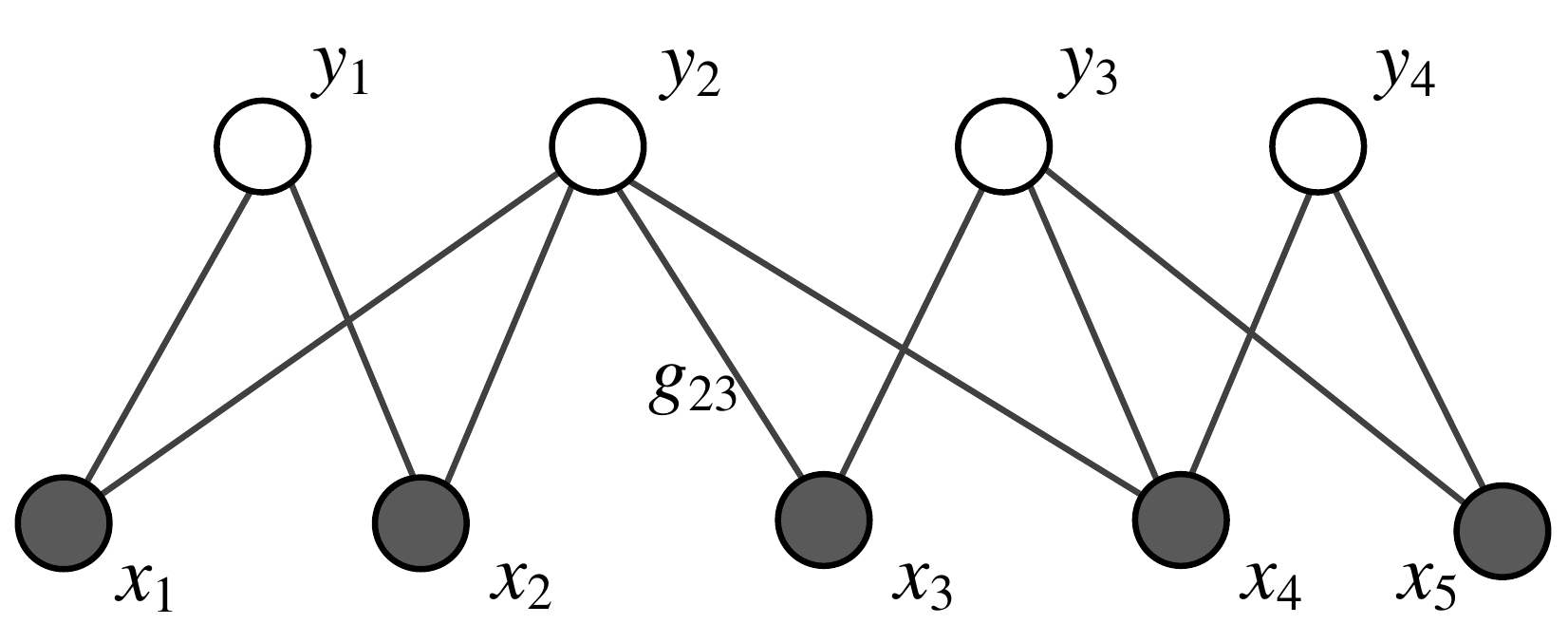}
\end{center}
\caption{A sample network scenario with number of SUs $m = 5$, number of PUs $n = 4$ and its bipartite
graph transformation. White circles represent PUs, black circles represent SUs and dashed lines illustrate SUs' monitor range.}
\label{fig:model}
\end{figure}

Consider {an} $m\times T$ matrix $\X$, which is the collection of $T$
realizations of vector $\xx$. The goal of \bica~is
to determine the distribution of the latent independent random variables $\yy$ and the
binary mixing matrix $\G$ from $\X$, such that $\X$ can be decomposed into the
mixing of realizations of $\yy$. From $\G$ and $\yy$, we can identify the PUs
and additionally infer
PUs' activities at different time slots.
Note that in (\ref{eq:boolean}), measurement noise is
not explicitly modeled, rather, is treated as independent sources.

%
%

\subsection{Why Binary Inference?}
\label{sec:linearvsbinary}
In this section, we motivate the use of a binary inference framework by
considering an alternative linear mixing model. In the linear model, at each
time slot, the received signal power at each monitor can be modeled as a linear
combination of the transmitted signal power from active PUs. More specifically, Let
$\vv = [v_1,v_2,\ldots,v_m]^T$, $\zz = [z_1,z_2,\ldots,z_n]^T$, and $\nn = [n_1, n_2, \ldots, n_m]^T$  be the random
vectors corresponding to the received, transmitted signal power and the Gaussian noise
{respectively}, and $\BH$ is the {$m \times n$} unknown channel gain matrix.
Both large-scale path loss with propagation loss factor $\alpha$ and small-scale
fading following the Rayleigh distribution are considered in this model.
The received signal power is a linear mixture of the transmitted signal power and the noise:
\beq
\vv = \BH \zz + \nn.
\label{eq:Linear}
\eeq

$n_i$ is a random variable with mean $\mu_i$.  If $\vv$ can
be observed directly, classical linear independent component analysis
(ICA)~\cite{Hyvarinen00} can be applied to determine $\BH$
and $\zz$. However, this method suffers from three problems. First, $z_i$'s
have to be non-Gaussian to be recoverable. Second, the channel gain matrix $\BH$
tends to vary over time.  Lastly, communicating the realizations of $\zz$ to a
centralized server for inference incurs higher overhead compared to its binary
counter part $\xx$. Furthermore, when only the binary quantized values of
$\vv$ are observable (i.e., $\xx$), ICA is not longer suitable due to
the non-linearity of the quantization function.

Next, we compare accuracy of the binary OR mixture model against a
quantized version of the above linear model, namely, $\xx' = U(\vv- \tau)$,
where $U(\cdot)$ is a step function defined as $U(r) = 1$ if $r > 0$, {$U(r) = 0$ otherwise}, and
$\tau$ is {a} pre-set threshold. We are interested to see the degree of
``information loss" due to the OR approximation.  In the simulation, 10
monitors and $n$ PUs are deployed in a 500x500 square meter area with $n$
varying from 5 to 20. Locations of PUs are chosen arbitrarily with the
restriction that no two PUs can be observed by the same set of monitors. The
PUs' transmit power levels are fixed at 20mW, the noise floor is -95dbm, and
the propagation loss factor is 3. The SNR detection threshold for {the} monitors is set
to be 5dB (above the noise floor). The value is chosen so that  the false alarm
probability (PU's are reported while none exists) is less than 10\%.
Elements of the binary mixing matrix $\G$ are either 1 or 0, depending on
the received signal for one respective PU only. {In other words, $g_{ij} = 1$
if the $i$th SU can detect transmissions from the $j$th PU}.  The PUs'
activities are modeled as a two-stage Markov chain with transition
probabilities uniformly distributed over (0,1).  We run the simulation for $T =
5,000$ slots and obtain the observations $\X$ and $\X_l$ for the binary OR and
linear mixing model, respectively. Figure~\ref{fig:error_comparison} shows the
{false alarm and miss detection probability}  in the binary OR model using the
results from the linear mixing model as the ground truth. From
Figure~\ref{fig:error_comparison}, we see the two models have very close
performance. We also experiment the case {in which} the initial phases of PU's
transmitted signal vary in $[0, 2\pi]$, and have similar observations.

\begin{figure}[tp]
\begin{center}
\includegraphics[width=3.0in]{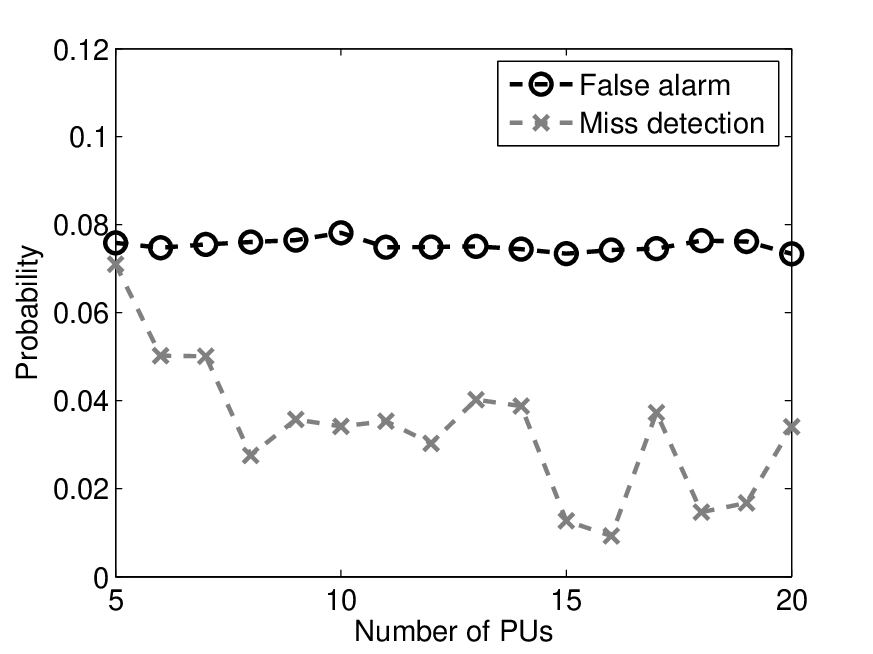}
\caption{False alarm and miss detection probability. Threshold $\tau$ is set at 5dB.}
\label{fig:error_comparison}
\end{center}
\end{figure}

%
%

To this end,  we conclude that the {\it binary model} is a good approximation
of the quantized linear model. As will be demonstrated later, efficient algorithms
can be devised for the PU separation problem under the binary model.

\section{Binary Inference Framework}
\label{sect:algo}
In this section, we first discuss the identifiability of binary independent
sources for the PUs from OR mixtures observed at the monitors, and then present an
inference algorithm that determines the unknown mixing matrix and underlying
sources.
\subsection{Identifiability}
For an $m$-dimension binary random vector $\xx$, the number of different
 realizations {is} $2^m$.  From the data matrix $\X$,
distribution of $\xx$ can be estimated in a non-biased manner as the number of
observations goes to infinity.  We can initialize $n = 2^m-1$ and $\G$ matrix
of dimension $m\times (2^m-1)$ with rows being all possible binary combinations
of length $n$ (with exception of the all-0 entry). This results in a complete
bipartite graph, {in which} an edge exists between any two vertices in $U$
{(the set of monitors)} and $V$ {(the set of signal sources)},
respectively. For a random variable $y_j\in V$, its neighbors in $U$ is given
by the non-zero entries in $g_{j}$ (i.e., the $j$th column of {matrix $\G$}).
Thus, at most $2^m - 1$ independent components can be identified.

Define {$p_l \stackrel{\Delta}{=} \MP(y_l = 1)$}. Let the set
$$\Y(\xx) = \left\{\yy \mid \bigvee_{j=1}^{n}{(g_{ij}\wedge y_j)} = x_i, \forall i\right\}.$$
Therefore,
\beq
\begin{array}{lll}
\MP(\xx) & = &\MP(\yy \in \Y(\xx)) = \sum_{\yy \in \Y(\xx)} \MP(\yy) \\
	 & = & \sum_{\yy \in \Y(\xx)} \prod_{l=1}^{2^m-1}{p_l^{y_l}(1-p_l)^{1-y_l},}
\end{array}
\label{eq:relation}
\eeq
where $\MP(\yy)$ is the joint probability of $\yy$. The last equality {holds} due to
the independence of $y_i$'s.

Given the distribution of random vectors $\xx\in \{0,1\}^m$,  $2^m - 1$
independent equations can be obtained from \eqref{eq:relation} due to the OR
mixture model in \eqref{eq:boolean}. {Since there are $2^m -1$ unknowns (i.e.,
$p_l, l = 1,\ldots, 2^m-1$), their values can be explicitly determined}.  Clearly,
ambiguity exists if two or more independent sources have the same set of
neighbors in $U$ (or equivalently, identical columns in $\G$). In this case,
binary information is insufficient to distinguish these sources.

The set of equations in \eqref{eq:relation} are polynomials of sum product forms,
which are difficult to solve.  This necessitates the design of specialized
algorithms. In the rest of the paper, abusing the notation a bit, we denote $\G$ the $m\times (2^m-1)$
adjacency matrix for the bipartite graph $G'=(U,V',E')$, {where $U=\set{x_1,x_2,\ldots,x_m}$
and $V'=\set{y_1,y_2,\ldots,y_{2^m-1}}$ (i.e. the set of all possible uniquely identifiable latent sources).}
Furthermore, we arrange $\G$ in the order
such that $g_{kl} = 1$ if $(l\ll k) = 1$, for $k = 0, \ldots, m-1$, where $\ll$
is the bit shift to the left. If the resulting $p_l = 0$ for some $l$, this
implies the corresponding column $g_{l}$ can be removed from $\G$.
For an example, consider the network scenario with $m = 3$, the initialized $\G$ matrix will
be:
$$
\G =
\begin{bmatrix}
1&0&1&0&1&0&1\\
0&1&1&0&0&1&1\\
0&0&0&1&1&1&1
\end{bmatrix},
$$
with rows corresponding to 3 monitors and columns corresponding to 7 identifiable
PUs.


\subsection{Inference algorithm}
Before proceeding to the details of the {proposed} algorithm, we first present a few
{related} technical lemmas {as follows}.
\begin{lem}
\label{lem:infer}
Consider a set  $\xx = [x_1, x_2, \ldots, x_{h-1}, x_h]^T$ generated by the data
model in \eqref{eq:boolean}, i.e., $\exists$ binary independent sources $\yy$,
s.t., $\xx = \G \otimes \yy$. The conditional random vector $\xx_{x_h = 0} = [x_1,
x_2, \ldots, x_{h-1}| x_h = 0]^T$ corresponds to the vector of the first $h-1$ elements of $\xx$ when $x_h =
0$. Then, $\xx_{x_h = 0} = \G '\otimes \yy'$, where $\G' = \G_{:,1, \ldots, 2^{h-1}}$
(i.e. the first $2^{h-1}$ columns of $\G$)
and $\MP(y'_l = 1) =  \MP(y_l = 1)$ for $l = 1, \ldots, 2^{h-1}$.
\end{lem}
\begin{proof}
We first derive the conditional probability distribution of the first $h-1$ observation variables
given $x_h = 0$,
\begin{equation}
\begin{array}{ll}
& \MP(x_1, x_2, \ldots, x_{h-1} \mid x_{h} = 0)  \\
= & \MP(x_1, x_2, \ldots, x_{h-1} \mid x_{h} = 0)\MP(x_{h}=0) \\
= & \hspace{-3mm}\displaystyle \sum_{\yy \in \Y(\xx)}{\prod_{l=1}^{2^h-1}{p_{l}^{y_{l}}(1-p_{l})^{1-y_{l}}}} \\
= & \hspace{-5mm}\displaystyle \sum_{\begin{array}{l} \yy_{1..2^{h-1}}\in \Y(\xx_{1..h-1})\\ y_{l} = 0, \forall g_{hl}=1 \end{array}}{\prod_{g_{hl} = 0}{p_{l}^{y_{l}}(1-p_{l})^{1-y_{l}}}\prod_{g_{hl}=1}{(1-p_{l})}}
\end{array}
\end{equation}

since $\MP(x_{h}=0) = \displaystyle \prod_{g_{hl}=1}{(1-p_l)}$, we have
\begin{equation}
\begin{array}{ll}
& \MP(x_1, x_2, \ldots, x_{h-1}\mid x_{h} = 0) \\
= & \displaystyle \sum_{\yy'\in \Y(\xx_{1:h-1})}{\prod_{l=1}^{2^h-1}{(p'_{l})^{y'_{l}}(1-p'_{l})^{1-y'_{l}}}} \\
= & \displaystyle \sum_{\begin{array}{l} \yy_{1, \ldots, 2^{h-1}}\in \Y(\xx_{1, \ldots, h-1})\\ y_{l} = 0, \forall g_{hl}=1 \end{array}}{\prod_{g_{hl} = 0}{p_{l}^{y_{l}}(1-p_{l})^{1-y_{l}}}.}
\end{array}
\end{equation}

Clearly, by setting $\MP(y'_l = 1) =  \MP(y_l = 1)$ for $l = 1, \ldots,
2^{h-1}$, the above equality holds.  In the other word, the conditional random
vector $\xx_{x_h = 0} = \G' \otimes \yy'$ for $\G' = \G_{:,1, \ldots, 2^{h-1}}$.
\end{proof}

The above lemma establishes that the
conditional random vector $\xx_{x_h = 0}$ can be represented as an OR mixing of
$2^{h-1}$ independent components. Furthermore, the set of the independent
components is the same as the first $2^{h-1}$ independent components of $\xx$
(under proper ordering).

Consider a sub-matrix of data matrix $\X$, $\X_{(h-1)\times T}^0$,
where the rows  correspond to observations of $x_1, x_2, \ldots, x_{h-1}$
for $t = 1, 2, \ldots, T$ such that $x_{ht} = 0$. Define $\X_{(h-1)\times T}$, which
consists of the first $h-1$ rows of $\X$.  Suppose that we have computed the \bica~
for data matrices $\X_{(h-1)\times T}^0$ and $\X_{(h-1)\times T}$. From
Lemma~\ref{lem:infer}, we know that $\X_{(h-1)\times T}^0$ is realization of OR mixtures of
independent components, denoted by $\yy_{2^{h-1}}^0$. Furthermore,
$\MP[\yy_{2^{h-1}}^0(l) = 1] =  \MP(y_l = 1)$ for $l = 1, \ldots, 2^{h-1}$.
Clearly, $\X_{(h-1)\times T}$ is realization of OR mixtures of $2^{h-1}$ independent
components, denoted by $\yy_{2^{h-1}}$.
Additionally, it is easy to see that the following holds:
$$
\begin{array}{ll}
& \MP[\yy_{2^(h-1)}(l)=1]\\
& = 1 - [1-\MP(\yy_{2^{h-1}}^0(l) = 1)][1-\MP(y_{l+2^{h-1}} = 1)] \\
& = 1 - (1-p_l)(1-p_{l+2^{h-1}}),
\end{array}
$$
where $l = 1, \ldots, 2^{h-1}$. Therefore,
\begin{equation}
\begin{array}{llll}
p_{l} & = & \MP(\yy^0_{2^{h-1}}(l)=1), & l = 1, \ldots, 2^{h-1}, \\
p_{l+2^{h-1}} & = & 1-\frac{1 - \MP(\yy_{2^{h-1}}(l)=1)}{1-\MP(\yy_{2^{h-1}}^0(l) = 1)}, & l = 2, \ldots, 2^{h-1}, \\
p_{2^{h-1} + 1} & = & \frac{\MF{(x_h=1 \wedge x_i=0, \forall i \in [1 \ldots h-1])}}
{\prod_{l=1 \ldots 2^h, l \neq 2^{h-1}-1}{(1-p_l)}}.
\end{array}
\label{eq:half2}
\end{equation}
The last equation above holds because realizations of $\xx$ where ($x_k=1$
while $x_i=0; \forall i \in \set{0,\ldots,k-1}$) are generated from OR mixtures of $\yy_{2^{k-1}}$'s only.
Define $\MF(A)$ as the frequency of event $A$.
To this end, we have the following iterative algorithm as illustrated in
Algorithm~\ref{algo:inc}.

When the number of monitors $m=1$, there are only two possible unique sources,
one that can be detected by the monitor, denoted by [1]; and one that cannot,
denoted by [0]. Their active probabilities can easily be calculated by counting
the frequency of $(x_1 = 1)$ and $(x_1 = 0)$ (lines~\ref{line1} --
\ref{line3}). If $m \ge 2$, we apply Lemma~\ref{lem:infer} and \eqref{eq:half2}
to estimate $p$ and $\G$ through a recursive process. We invoke {\sc FindBICA} on
two sub-matrices $\X^0_{(m-1)\times T}$ and $\X_{(m-1)\times T}$ computed from
$\X$ to determine $p_{1 \ldots 2^{m-1}}$ and $p'_{1 \ldots 2^{m-1}}$, then
infer $p_{2^{m-1}+1 \ldots 2^{m}}$ as in (\ref{eq:half2}) (lines~\ref{line6} --
\ref{line8}). Finally, $p_h$ and its corresponding column $g_h$ in $\G$ are
pruned in the final result if $p_h < \varepsilon$ (lines~\ref{line9} --
\ref{line11}).

\begin{algorithm}
\caption{Incremental binary ICA inference algorithm}
\small
\label{algo:inc}
\SetKwData{Left}{left}
\SetKwInOut{Input}{input}
\SetKwInOut{Output}{output}
\SetKwInOut{Init}{init}
\SetKwFor{For}{for}{do}{endfor}
\SetKwFunction{FindBICA}{FindBICA}
\FindBICA($\X$)\\
\Input{Data matrix $\X_{m\times T}$}
\Init{$n = 2^{m} -1$\; $p_h = 0, h = 1, \ldots, n$\;
$\G$ = $m \times (2^{m} -1)$ matrix with rows corresponding all possible binary vectors of length $m$\;
$\varepsilon$ = the minimum threshold for $p_h$ to be considered a real component;} \BlankLine

\nl\eIf{$m = 1$} {\label{line1}
\nl    $p_1 = \MF(x_1 = 0)$\;\label{line2}
\nl    $p_2 = \MF(x_1 = 1)$\;\label{line3}
    } {
\nl    $p_{1 \ldots 2^{m-1}}$ = \FindBICA($\X^0_{(m-1)\times T}$)\;\label{line4}
\nl    $p'_{1 \ldots 2^{m-1}}$ = \FindBICA($\X_{(m-1)\times T}$)\;\label{line5}
\nl    \For{$l = 2, \ldots, 2^{m-1}$}{\label{line6}
\nl        $p_{l+2^{m-1}} = 1-\frac{1-p'_l}{1-p_l}$\;\label{line7}
       }
\nl    $p_{2^{m-1}+1} = \frac{\MF{(x_m=1 \wedge x_i=0, \forall i \in [1 \ldots m-1])}}{\prod_{l=1 \ldots 2^m-1, l \neq 2^{m-1}+1}{(1-p_l)}}$\;\label{line8}
}
\nl\For{$h = 1, \ldots, 2^{m}$}{\label{line9}
\nl    \If{$(p_h < \varepsilon) \vee (p_h = 0)$}{\label{line10}
\nl            prune $p_h$ and corresponding columns $g_{h}$\;\label{line11}
        }
}
\nl\Output{$p$ and $\G$}
\end{algorithm}
\normalsize

\paragraph*{Computation complexity}
Let $S(m)$ be the computation time for finding \bica~given $\X_{m\times T}$. From
Algorithm~\ref{algo:inc}, we have,
\beq
S(m) = 2S(m-1) + c2^m,
\eeq
where $c$ is a
constant. It is easy to verify $S(m) = cm 2^m$.  Therefore,
Algorithm~\ref{algo:inc} has an exponential computation complexity with respect
to $m$. This is clearly undesirable for large $m$'s. However, we notice that in
practice, correlations among $x_i$'s exhibit locality, and {matrix $\G$} tends
to be sparse. Instead of using a complete bipartite graph to represent $\G$, the
degree of vertices in $V$ (or the number of non-zero elements in $g_{k}$)
tends to be much less than $m$. More specifically, for every pair {of monitors} $i$ and $k$,
we compute {the covariance between their observations:}

\beq
cov(i,k) = \frac{\sum_{t}^T{x_{it}x_{kt}}}{T} -
\frac{\sum_{t}^T{x_{it}}}{T}\frac{\sum_{t}^T{x_{kt}}}{T}.
\label{eq:cov}
\eeq

If $cov(i,k) < \epsilon$, where $\epsilon$ is a small value (e.g., the upper confidence bound of
$cov(i,k)$ estimate), we can remove the corresponding columns in $\G$ and
elements in $\yy$.

\section{The Inverse Problem}
\label{sec:inverse}
Now we have the mixing matrix $\G_{m\times n}$ and the active probabilities
$\MP(\yy)$, given observation $\X_{m\times T}$, the inverse problem concerns
inferring the realizations of the latent variables $\Y_{n\times T}$. Extracting multiple
PUs' activities from the OR mixture observations is a challenging but important
problem in cognitive radio networks. Interesting information, such
as the PU channel usage pattern can be inferred once $\Y$ is available. The
SUs will then be able to adopt better spectrum sensing and access strategies to
exploit the spectrum holes more effectively.

{Recall that $n$ is the number of PUs (latent variables).}
Denote $y_i$ a binary variable for the
$i$th latent variable. Let $\xx = \G \otimes \yy$.  We assume that the probability of
observing $\X$ given $\xx$ depends on their Hamming distance $d(\xx, \X) =
\sum_{i}{|\X_i-x_i|}$, and $\MP(\xx|\X) = p_e^{d(\xx, \X)}(1-p_e)^{m-d(\xx, \X)}$, where $p_e$ is the
error probability of the binary symmetric channel. To determine $\yy$, we
can maximize the posterior probability of $\yy$ given $\X$ derived as follows,
$$
\begin{array}{lll}
\MP\{\yy|\X\} & = & \frac{\MP\{\X|\yy\}\MP\{\yy\}}{\MP\{\X\}} \\
         & = & \frac{\MP\{\X|\yy\}\MP\{\yy\}}{\MP\{\X\}} \\
         & \stackrel{(a)}{=} & \frac{\MP\{\X,\xx|\yy\}\MP\{\yy\}}{\MP\{\X\}} \\
         & \stackrel{(b)}{=} & \frac{\MP\{\X|\xx\}\MP\{\yy\}}{\MP\{\X\}} \\
         & = & \frac{\prod_{i=1}^{m}{\MP\{\X_i|x_i\}}\prod_{j=1}^{n}{\MP\{y_i\}}}{\MP\{\X\}} \\
         & = & \frac{\prod_{i=1}^{m}{ p_e^{|x_i- \X|}(1-p_e)^{1-|x_i-\X|}}\prod_{j=1}^{n}{p_i^{y_i}(1-p_i)^{1-y_i}}}{\MP\{\X\}}
\end{array}
$$
where $\xx = \G \otimes \yy$.  $(a)$ and $(b)$ are due to the deterministic
relationship between $\xx$ and $\yy$.  
{Recall that $x_i = \bigvee_{j=1}^{n}{(g_{ij}\wedge y_j)}, i = 1, \ldots, m$.
With $M$ is a ``large enough'' constant, we can use big-$M$ formulation}~\cite{Griva2008Linear}
{to relax the disjunctive set and convert the above relationship between $\xx$ and
$\yy$ into the following two sets of conditions:}
\begin{equation}
\begin{array}{llll}
 \displaystyle x_i & \le & \sum_{j=1}^{n}{g_{ij}y_{j}}, \mbox{ } \forall i=1,\ldots,m. \\
 \displaystyle M\cdot x_i & \ge & \sum_{j=1}^{n}{g_{ij}y_{j}}, \mbox{ } \forall i=1,\ldots,m. \\
\end{array}
\end{equation}
Here, since $\sum_{j=1}^{n}{g_{ij}y_{j}} \le n$, we can set $M = n$.
Finally, taking $\log$ on both sides  and introducing additional auxiliary variable
$\alpha_i = |\X_i - x_i|$, we {have} the the following integer programming
problem:
\begin{equation}
\begin{array}{lll}
\underset{\alpha,y}{\max.}& \displaystyle \sum_{i=1}^m{\left[\alpha_i\log{p_e} + (1-\alpha_i)\log(1-p_e)\right]}\\
    & + \sum_{j=1}^n{\left[(1-y_j)\log{(1-p_j)} + y_j\log{p_j}\right]} \\
\mbox{s.t.}& \displaystyle x_i \le \sum_{j=1}^{n}{g_{ij}y_{j}}, \hspace{10mm} \forall i=1,\ldots,m, \\
    & \displaystyle n\cdot x_i \ge \sum_{j=1}^{n}{g_{ij}y_{j}}, \hspace{5.5mm} \forall i=1,\ldots,m, \\
    & \alpha_i \ge \X_i - x_i, \hspace{11mm} \forall i=1,\ldots,m, \\
    & \alpha_i \ge x_i - \X_i, \hspace{11mm} \forall i=1,\ldots,m, \\
    & \alpha_i, x_i, y_j = \set{0,1}, \hspace{4.5mm} \forall i=1,\ldots,m, j=1,\ldots, n. \\
\end{array}
\label{eq:IPP}
\end{equation}
This optimization function can be solved using ILP solvers. Note that $p_e$ can
be thought of the penalty for mismatches between $x_i$ and $\X_i$.
\normalsize
\paragraph*{Zero Error Case}
If $\X$ is perfectly observed, containing no noise,
we have $p_e = 0$ and $\alpha_i = \xx_i - \X_i = 0$, or equivalently,
$\xx_i = \X_i$. The integer programming problem in \eqref{eq:IPP}
can now be simplified as:
\begin{equation}
\begin{array}{lll}
\underset{y}{\max.}& \displaystyle \sum_{j=1}^n{\left[(1-y_j)\log{(1-p_j)} + y_j\log{p_j}\right]} \\
\mbox{s.t.}& \displaystyle \X_i \le \sum_{j=1}^{n}{g_{ij}y_{j}}, \forall i=1,\ldots,m, \\
    & \displaystyle n\cdot \X_i \ge \sum_{j=1}^{n}{g_{ij}y_{j}}, \forall i=1,\ldots,m, \\
    & y_j = \set{0,1}, \forall j=1,\ldots, n. \\
\end{array}
\label{eq:IPPSim}
\end{equation}

Clearly, the computation complexity of the zero error case is lower
compared to \eqref{eq:IPP}. It can also be used in the case where
prior knowledge regarding the noise level is not available.

\section{Evaluation}
\label{sect:eval}

In this section, we first introduce the performance metrics, and then present
evaluation results on a synthetic data set varying the number of PUs.

\subsection{Performance metrics}
We denote by $\hat{p}$ and $\hat{\G}$ the inferred active probability of
PUs and the inferred mixing matrix, respectively.
\subsubsection{Structure Error Ratio}
This metric indicates how accurate the mixing matrix is estimated. It is defined by the Hamming
distance between $\G$ and $\hat{\G}$ divided by its size.
\beq
\begin{array}{lll}
\bar{H}_g & \stackrel{\Delta}{=} & \frac{1}{mn} \sum_{i=1}^{n} d^H(g_{i},\hat{g}_{i}).
\end{array}
\eeq
To estimate $\bar{H}_g$ however, two challenges remain:
First, the number of inferred independent
components may not be identical as the ground truth. Second, the order of
independent components in $\G$ and $\hat{\G}$ may be different.

To solve the first problem, we can either prune $\hat{\G}$ or introduce columns into $\G$
to equalize the number of components ($n = \hat{n}$, where $\hat{n}$ is the number of columns
in $\hat{\G}$). For the second problem, we propose a matching algorithm that minimizes the
Hamming distance between $\G$ and $\hat{\G}$ by permuting the {corresponding} columns in $\hat{\G}$.

\paragraph*{Structure Matching Problem}
A naive matching algorithm needs to consider all $\hat{n}!$
column permutations of $\hat{\G}$, and chooses the one that has the minimal Hamming
distance to $\G$. This approach incurs an exponential computation complexity. Next,
we first formulate the best match as an ILP problem. Denote the
Hamming distance between column $\hat{g}_{i}$ and $g_{j}$ as
$c_{ij} \geq 0$. Define a permutation matrix $\A_{n \times n}$ with $a_{ij} = 1$
indicating that the $i$th column in $\hat{\G}$ is matched with the $j$th column in $\G$.
The problem now is to find a permutation matrix such that the total
Hamming distance between $\G$ and $\hat{\G}$ (denoted by $d^H(\G,\hat{\G})$)
is minimized. We can formulate this problem as an ILP as follows:
\begin{equation}
\begin{array}{ll}
\underset{a}{\min.}& \displaystyle \sum_{i=1}^n{\sum_{j=1}^n{c_{ij} a_{ij}}} \\
\mbox{s.t.}& \displaystyle \sum_{i=1}^n{a_{ij}} = 1, \\
    & \displaystyle \sum_{j=1}^n{a_{ij}} = 1, \\
    & a_{ij} = 0,1 \mbox{ } \forall i,j=1, \ldots ,n.
\end{array}
\label{eq:pgmatching}
\end{equation}
\begin{algorithm}[tp]
\caption{Bipartite graph matching algorithm} \small
\label{algo:pg} \SetKwData{Left}{left} \SetKwInOut{Input}{input}
\SetKwInOut{Output}{output} \SetKwInOut{Init}{init}
\SetKwFor{For}{for}{do}{endfor} \SetKwFunction{Hungarian}{BipartiteMatching}
\SetKwFunction{MatchPG}{MatchPG}
\MatchPG($\G$, $\hat{\G}$, $\hat{p}$)\\
\Input{$\G_{m \times n}$, $\hat{\G}_{m \times \hat{n}}$, $\hat{p}_{1 \times \hat{n}}$; ($n \le \hat{n} \le 2^{m}$)}
\Init{$\hat{\G}'_{m \times \hat{n}} = 0$; $\hat{p}'_{1 \times \hat{n}} = 0$; $\C_{\hat{n} \times \hat{n}} = 0$\;}
\BlankLine
\nl\For{$i = 1, \ldots, \hat{n}$}{\label{Line1}
\nl    \For{$j = 1, \ldots, \hat{n}$}{\label{Line2}
\nl        \eIf{$g_{i} = 0$} {\label{Line3}
\nl            $c_{ij} = d^H(g_{i}, \hat{g}_{j}) \times m$\;\label{Line4}
        } {
\nl            $c_{ij} = d^H(g_{i}, \hat{g}_{j})$\;\label{Line5}
        }
    }
}
\nl$\A$ = \Hungarian($\C$)\;\label{Line6}
\nl\For{$i = 1, \ldots, \hat{n}$}{\label{Line7}
\nl    find $j$ such that $a_{ij} = 1$\;\label{Line8}
\nl    $\hat{g}'_{i} = \hat{g}_{j}$\;\label{Line9}
\nl    $\hat{p}'_i = \hat{p}_j$;\label{Line10}
}
\nl Prune $\hat{\G}'$: $\hat{\G}' = \hat{g}'_{1 \ldots n}$\;\label{Line11}
\nl Prune $\hat{p}'$: $\hat{p}' = \hat{p}'_{1 \ldots n}$\;\label{Line12}
\nl\Output{$\hat{\G}'$ and $\hat{p}'$}\label{Line13}
\end{algorithm}
\normalsize

The constraints ensure the resulting $\A$ is a permutation matrix. This problem
 can be solved using ILP solvers. However, we observe that the ILP is equivalent to
a maximum-weight bipartite matching problem. In the bipartite graph, the vertices
are positions of the columns, and the edge weights are the Hamming distance of the
respective columns. If we consider $d^H(g_i,\hat{g}_i)$, the Hamming distance between column $g_i$
and $\hat{g}_i$, to be the ``cost'' of matching $\hat{g}_i$
to $g_i$, then the maximum-weight bipartite matching problem can be solved
in $O(n^3)$ running time~\cite{KuhnHungarian}, where $n$ is
the number of vertices. The algorithm requires $\G$ and $\hat{\G}$ to have the same
number of columns.

One {greedy} solution is to prune $\hat{\G}$ by selecting the top $n$ components from
$\hat{\G}$, which have the highest associated probabilities $\hat{p}_{i}$ since
they are the most likely true components. However, when $T$ is small and/or under
large noise, we may not have sufficient observations to correctly identify
components in $\G$ with high confidence. As a result, true components might
have lower active probabilities comparing to the noise components. To address
the problem, we instead keep a larger $\hat{n}$ and introduce $\hat{n} - n$
artificial components into $\G$. These components will be represented by zero
columns in $\G$.
While matching the inferred columns in $\hat{\G}$ to the columns in $\G$, clearly an
undesirable scenario occurs when we accidently match a column in $\hat{\G}$ to an
additive zero column in $\G$. This happens when an inferred column $\hat{g}_{i}$ is sparse
(i.e. having a very small Hamming distance to the zero column). To avoid the
incident, we multiply the cost of matching any column in $\hat{\G}$ to
a zero column in $\G$ by $m$. This eliminates the case {in which} a column $\hat{g_i}$
is matched with a zero column in $\G$, since it is more expensive than matching
with another non-zero column $g_i$.
We can now select the best $n$ candidates in $\hat{\G}$, which yields
a reduced mixing matrix $\hat{\G}'$ of size $m \times n$, and elements in active
probability vector $\hat{p}'$ will also be selected accordingly. The solution
to the structure matching problem is detailed in Algorithm \ref{algo:pg}.

\begin{figure}[tp]
\begin{center}
\hspace{-0.4in}\includegraphics[width=3.2in]{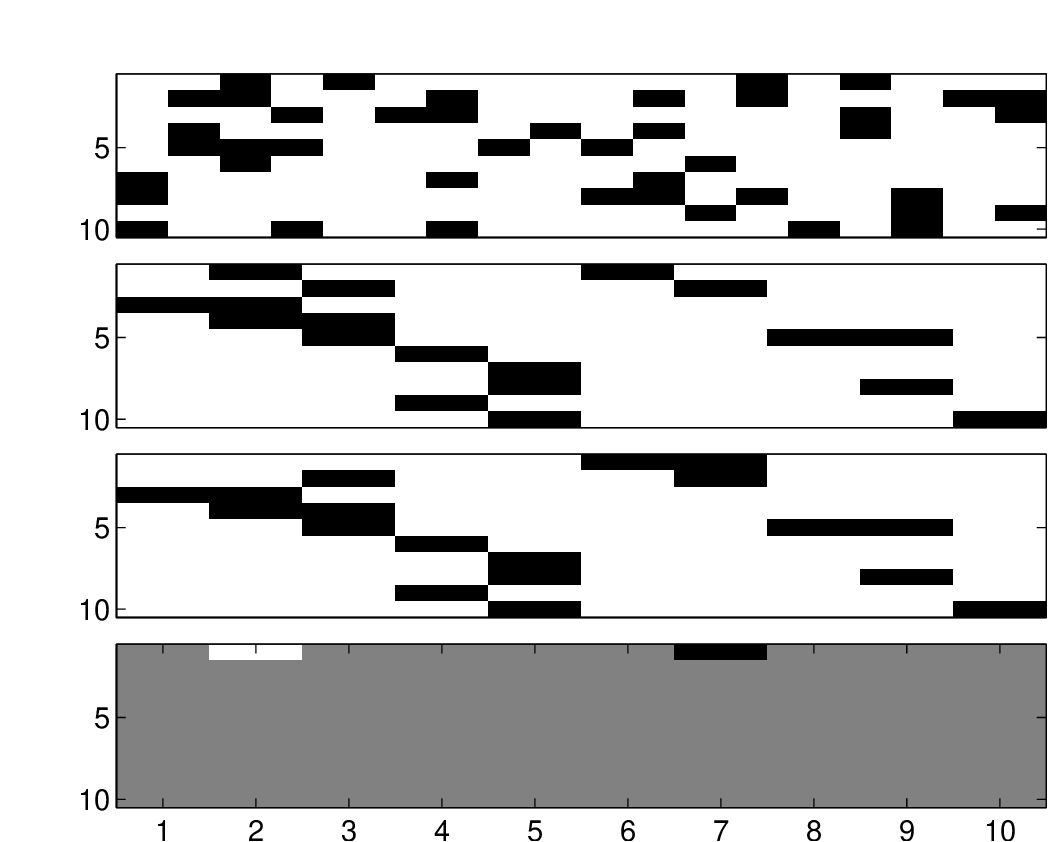}
\end{center}
\caption{From top to bottom: inferred matrix $\hat{\G}$ with 18
inferred components, transformed matrix $\hat{\G}'$ with only 10
components remaining (8 noisy ones were removed), original $\G$
and difference matrix between $\G$ and $\hat{\G}'$. Black dot = 1
and white dot = 0.} \label{fig:illus}
\end{figure}

In the algorithm,  lines~\ref{Line1} -- \ref{Line5} build the input weight matrix
$\C_{\hat{n} \times \hat{n}}$ for the bipartite matching algorithm. If $g_i$ is
a zero column, $c_{ij}$ will be scaled by $m$ to avoid the matching
between column $g_i$ and $\hat{g}_j$ (line~\ref{Line4}). The bipartite matching algorithm
finds the optimal permutation matrix $\A$
to transform $\hat{\G}$ into $\hat{\G}'$ that is ``closest'' to $\G$ (lines
\ref{Line6} -- \ref{Line10}). We are only interested in the first $n$ columns
of $\hat{\G}'$ and $\hat{p}'$ (as they most likely represent the true PUs).
Therefore, $\hat{\G}'$ and $\hat{p}'$ are pruned in lines~\ref{Line11} -- \ref{Line13}.

%
%
%
%

As an example, the inferred result of a random network with $n = m = 10$ is given in Figure~\ref{fig:illus}.
Non-zero entries and zero entries of $\G$,
$\hat{\G}$, and $\hat{\G}'$ are shown as black and white dots,
respectively. The entry-wise difference matrix $|\G-\hat{\G}'|$ is
given in the bottom graph. {Gray dots in the difference matrix
indicate identical entries in the inferred $\hat{\G}$ and the original
$\G$}; and black dots indicate different entries (and thus errors in the inferred matrix).
In this case, only the first row
(corresponding to the first monitor $x_1$) contains some errors.
\subsubsection{PU Active Probability Error Ratio}
Prediction error in the inferred active probabilities of PUs is
measured by the root mean square error ratio between $\hat{p}'$ and $p$, defined as {follows}:
\beq
\begin{array}{lll}
\bar{P} & \stackrel{\Delta}{=} & \sqrt{\frac{\sum_{i=1}^{n}
(\hat{p}'_i - p_i)^2}{n}} / \frac{\sum_{i=1}^{n} p_i}{n}.
\end{array}
\eeq
PU active probability error ratio can be interpreted as the fraction of the inferred active probability
that deviates from the true values.

\subsubsection{Miscount in the number of PUs}
Accuracy of the inference algorithm can also be assessed by evaluating the
difference between the number of inferred PUs $\hat{n}$ and the real number of
PUs $n$ in the system.  Clearly, with a smaller threshold value $\varepsilon$ in
Algorithm~\ref{algo:inc}, the number of inferred PUs $\hat{n}$ may increase. In
the subsequent experiments, we fix $\varepsilon = 0.01$ and evaluate the changes
in $\hat{n}$ as the real number of PUs increases from 5 to 20.

\subsubsection{PU Activity Error Ratio}
After applying {\sc FindBICA} in Algorithm~\ref{algo:inc} on the measurement data of length $T$ to obtain
$\hat{\G}$ and $\hat{p}$, realizations of the hidden causes (i.e. PUs)
can be computed by solving the {maximum likelihood estimation} problem in (\ref{eq:IPP}). We define
\beq
\begin{array}{lll}
\bar{H}_y & \stackrel{\Delta}{=} & \frac{1}{nT} \sum_{i=1}^{T} d^H(y_{i},\hat{y}_{i}),
\end{array}
\eeq
where $y_{i}$ is the $i$'th column of $\Y$.
Similar to $\bar{H}_g$, this metric measures how accurately the PUs' activity
matrix is inferred by calculating the ratio between the size of $\yy$ and the absolute
difference between $\yy$ and $\hat{\yy}$.

\subsection{Experimental results}

For evaluation, we consider the same simulation scenario as in
Section~\ref{sec:linearvsbinary}, where 10 monitors are randomly deployed to
monitor $n$ PUs {using} a single channel.
Algorithms are implemented in Matlab, and all experiments are conducted on
a workstation with an Intel Core 2 Duo T5750\makeatletter@\makeatother2.00GHz
processor and 2GB RAM. Noise is introduced by randomly flip a bit in the
observation matrix $\X$ from 1 to 0 (and vice versa) at probability $e$.  $e$
is set at 0\%, 2\%, and 5\% in our simulations. All presented results are
averages of 50 runs with different initializations.  In the experiments, we
vary the number of PUs $n$ and the number of observations $T$.
\vspace{2mm}

\subsubsection{Varying the number of PUs}

\begin{figure*}[tp]
\begin{center}
\begin{tabular}{ccc}
\hspace{-0.15in}\includegraphics[width=2.7in]{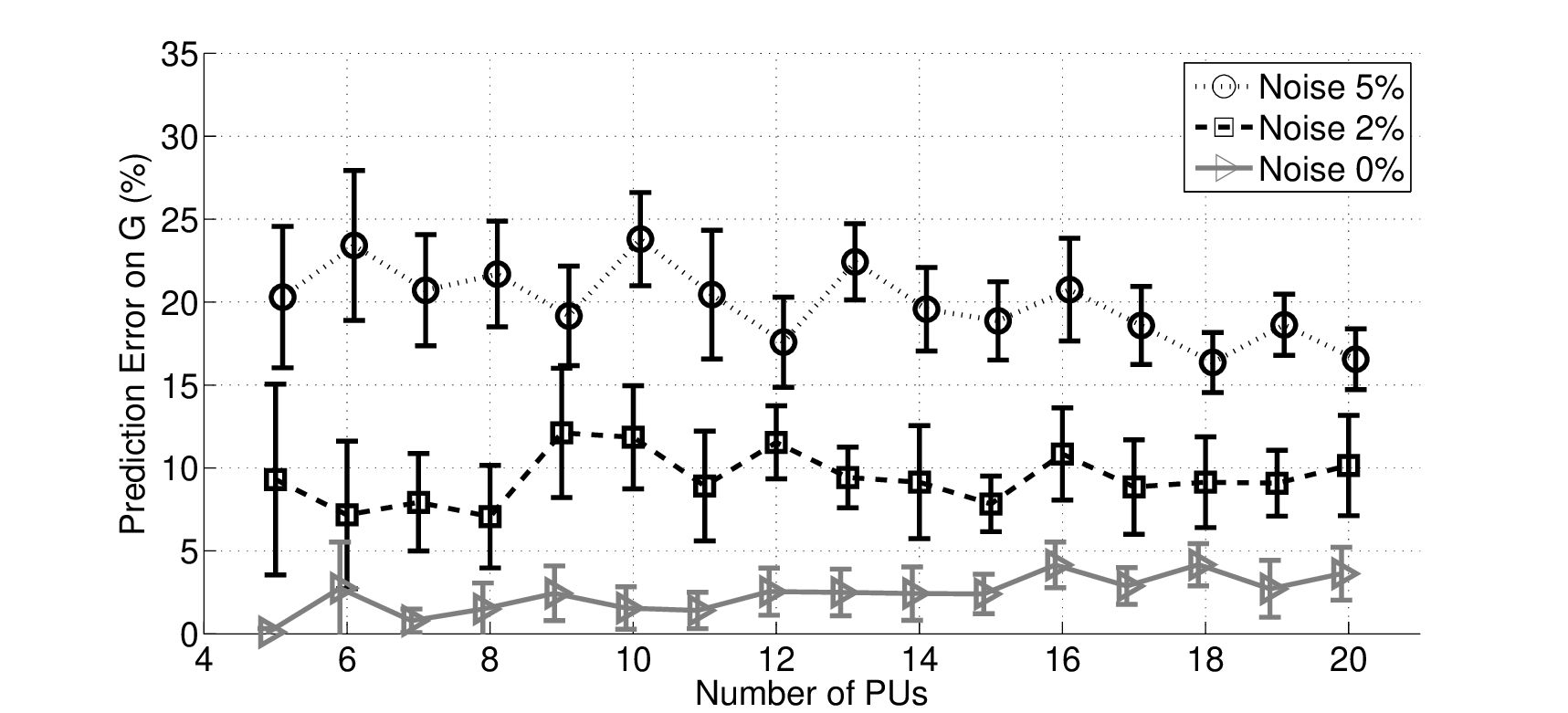}
\hspace{-0.3in} & \hspace{-0.3in} \includegraphics[width=2.7in]{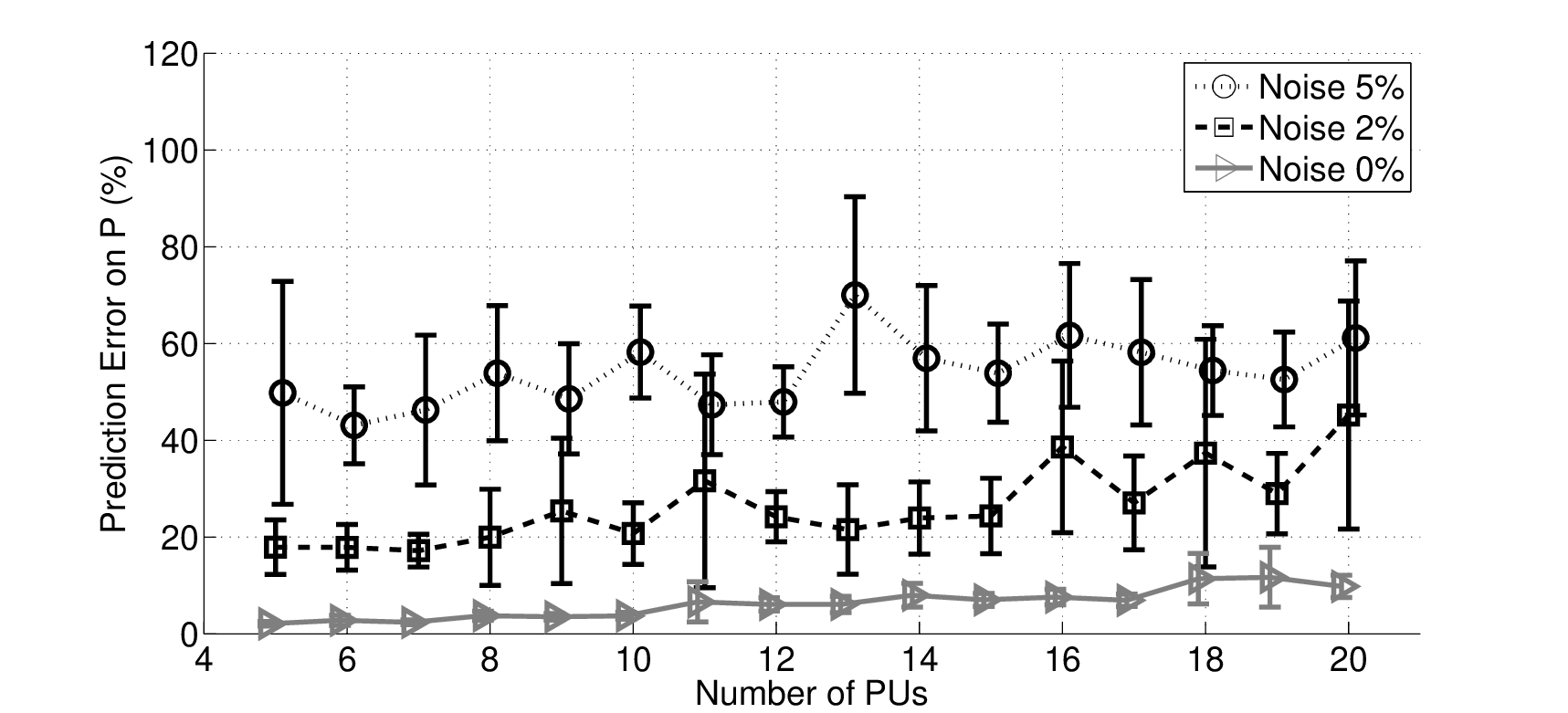}
\hspace{-0.3in} & \hspace{-0.3in} \includegraphics[width=2.7in]{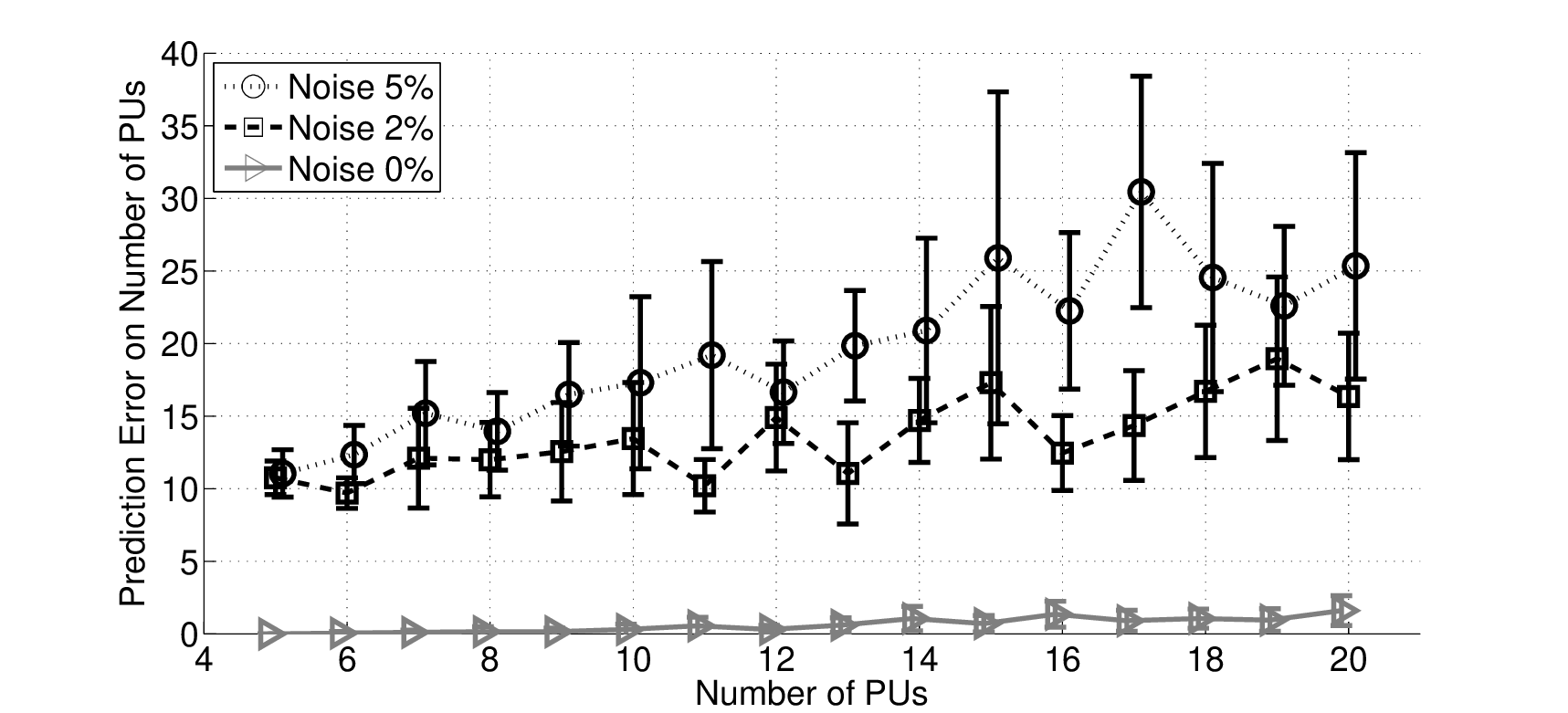}\\
\hspace{-0.15in}(a) Structure error ratio
\hspace{-0.3in} & \hspace{-0.3in} (b) PU active probability error ratio
\hspace{-0.3in} & \hspace{-0.3in} (b) Miscount of PUs \\
\end{tabular}
\begin{tabular}{cc}
\includegraphics[width=2.7in]{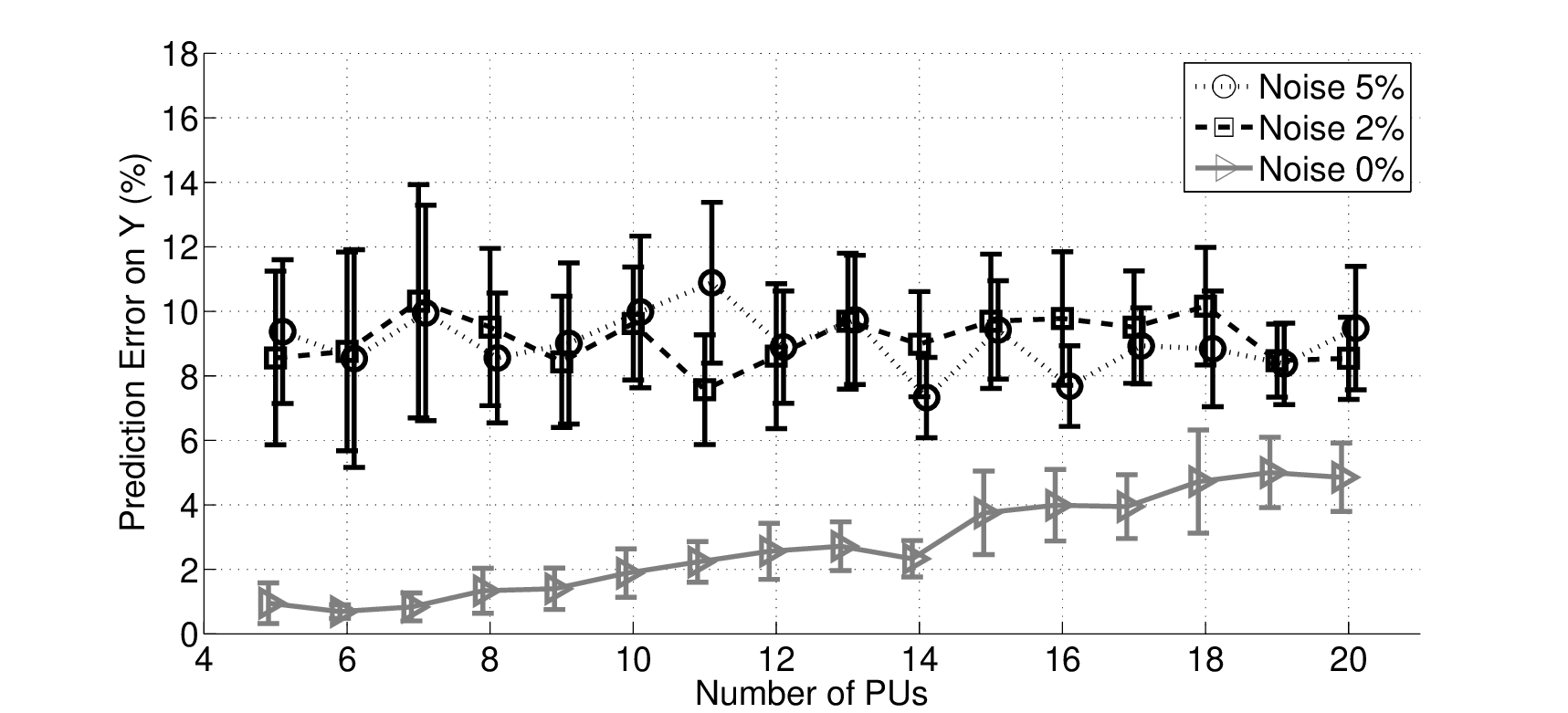}
\hspace{-0.3in} & \hspace{-0.3in} \includegraphics[width=2.7in]{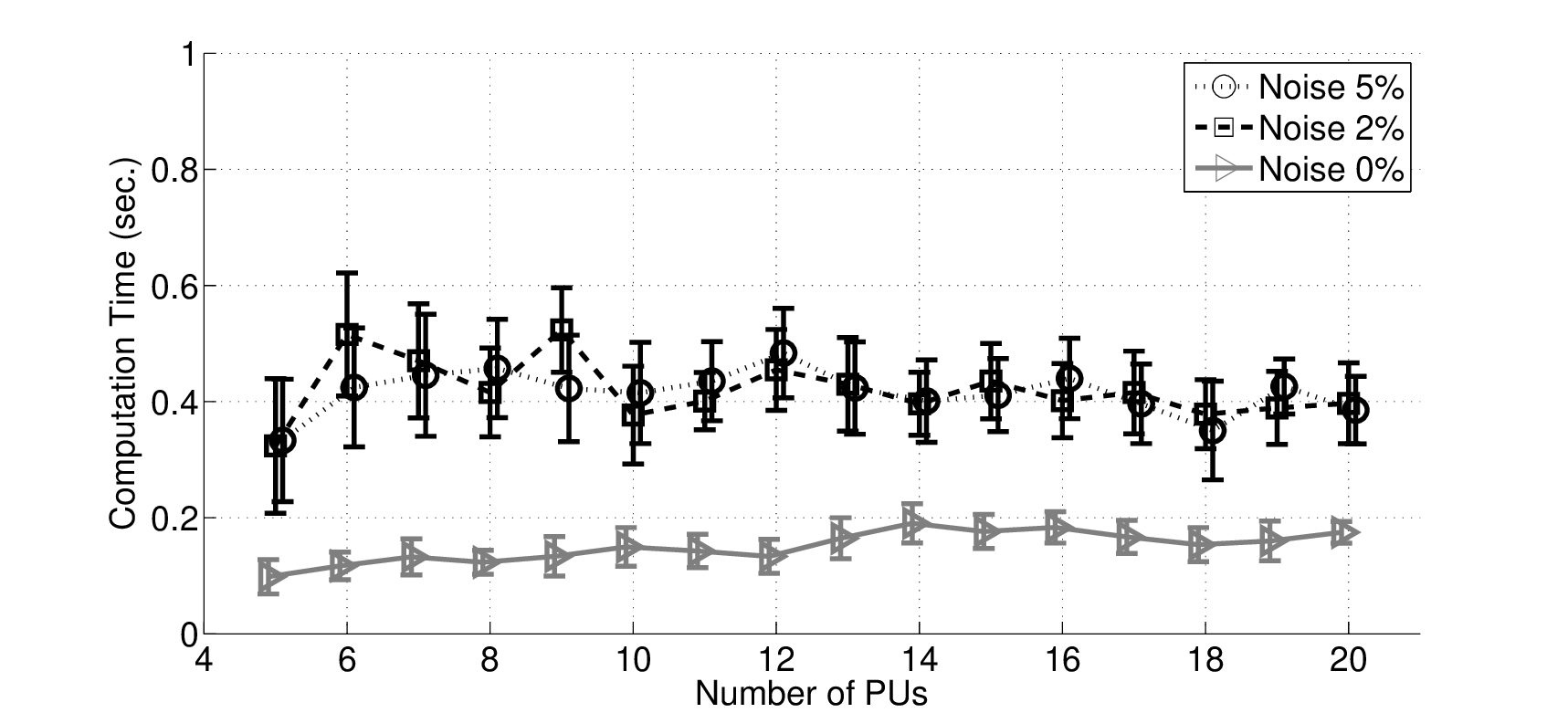}\\
(d) PU activity error ratio
\hspace{-0.3in} & \hspace{-0.3in} (e) Computation time\\
\end{tabular}
\end{center}
\caption{
Effects of noise level and the number of PUs on inference results. Noise is
introduced at three levels: 0\%, 2\% and 5\%. Results are averages of 50 runs
with different initial seed. Symmetric error bars indicate standard
deviations.}
\label{fig:error}
\end{figure*}

{In the first set of experiments, we fix the sample size $T = 10,000$ and
vary the number of PUs from 5 to 20 to study its impact on the accuracy of our
method.  Experiment results over
50 runs for each PU setting are shown in Figure}~\ref{fig:error}.  In absence of
noise, we observe that the inferred mixing matrix $\hat{\G}$ is mostly correct
even for a large number of PUs (Figure~ \ref{fig:error}(a)). As the number of
PUs increases, errors in the inferred active probabilities and the inferred
number of PUs tend to increase though within 10\% and 0.65 as shown
in Figure~\ref{fig:error}(b) and (c). Recall that the PUs are ordered based on
their respective columns in $\hat{\G}$. Therefore, errors in $\hat{\G}$ may
have a cascading effect on PU active probability error ratio since we may
compare the wrong pair of PUs as a result.  Performance of the algorithm tends
to degrade with more noises. In particular, as shown in
Figure~\ref{fig:error}(c), more components are inferred compared to the ground
truth. This is because when the noise probability $p_e$ is greater than the
threshold value, some noise components are erroneously introduced.

The errors in determining the set of active PUs are shown in Figure~\ref{fig:error}(d).
We can see again that the proposed algorithm achieves remarkable accuracy at zero-noise
level. Prediction error on $\Y$ is only 1\% for 5 PUs, and gradually increases
up to 5\% for 20 PUs. At 2\% and 5\% noise levels, performance degrades as the
prediction error goes up to about 10\%. Noise has two effects on the solution to
the inverse problem. First, the inferred mixing matrix and active probability
can be erroneous. Second, no maximum likelihood estimator guarantees to give
the exact result when the problem is under or close to being under-determined with
noisy measurements. In fact, we have experimented with the case {in which}
$\G$ and $p$ are both known, the results are similar. This implies that the main
source of errors in solving the inverse problem comes from the problem itself being
under-determined or close to being under-determined.

Finally, as shown in Figure~\ref{fig:error}(e), the computation time of proposed algorithms is
negligible, mostly under 0.2 second without noise. With noise,
computation time increases but is similar at the two noise levels (differing by less
than 0.5 second). 
The presence of noise may introduce noise components and render the
estimation of correlation (in Equation (\ref{eq:cov})) inaccurate. Thus, higher processing time entails.

\vspace{1mm}

\begin{figure*}[tp]
\begin{center}
\begin{tabular}{ccc}
\hspace{-0.15in}\includegraphics[width=2.7in]{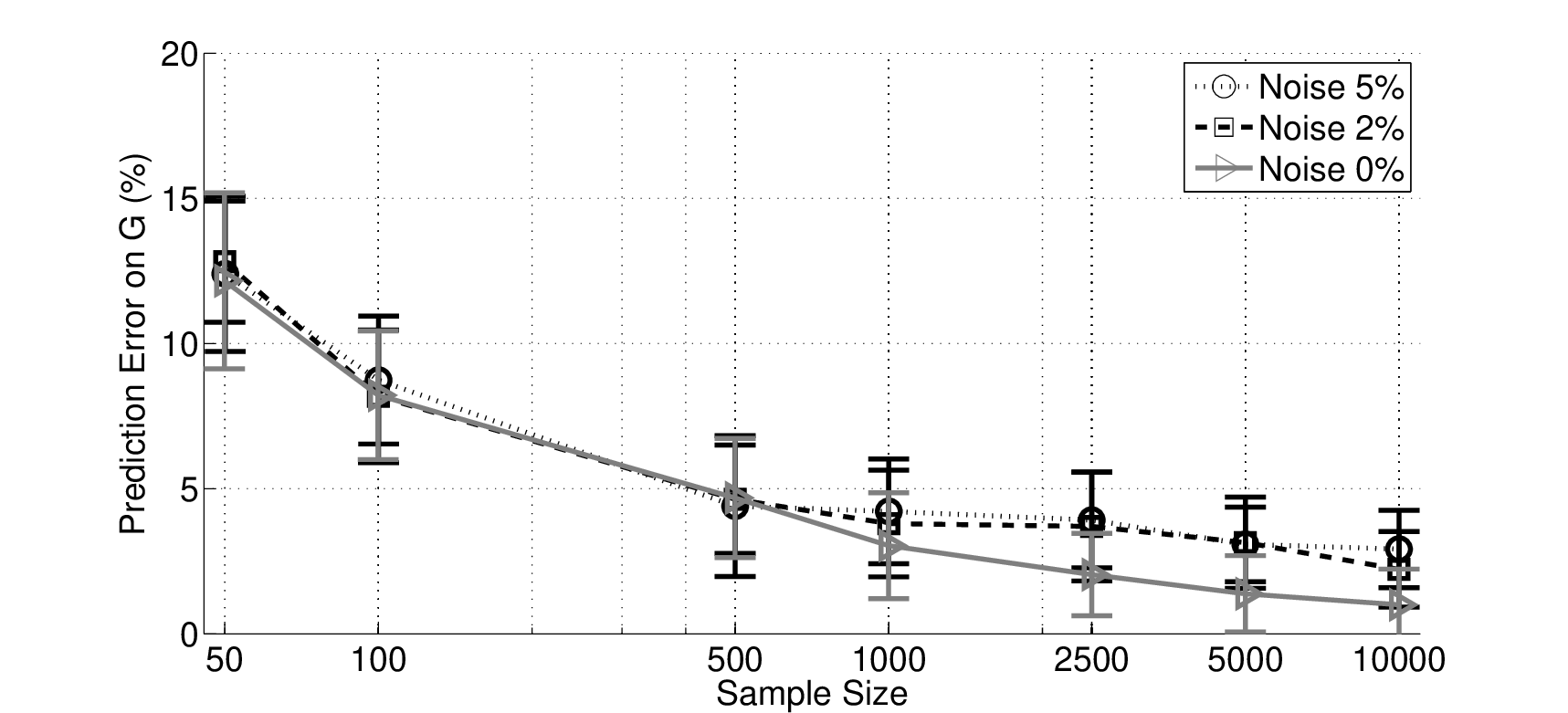}
\hspace{-0.3in} & \hspace{-0.3in} \includegraphics[width=2.7in]{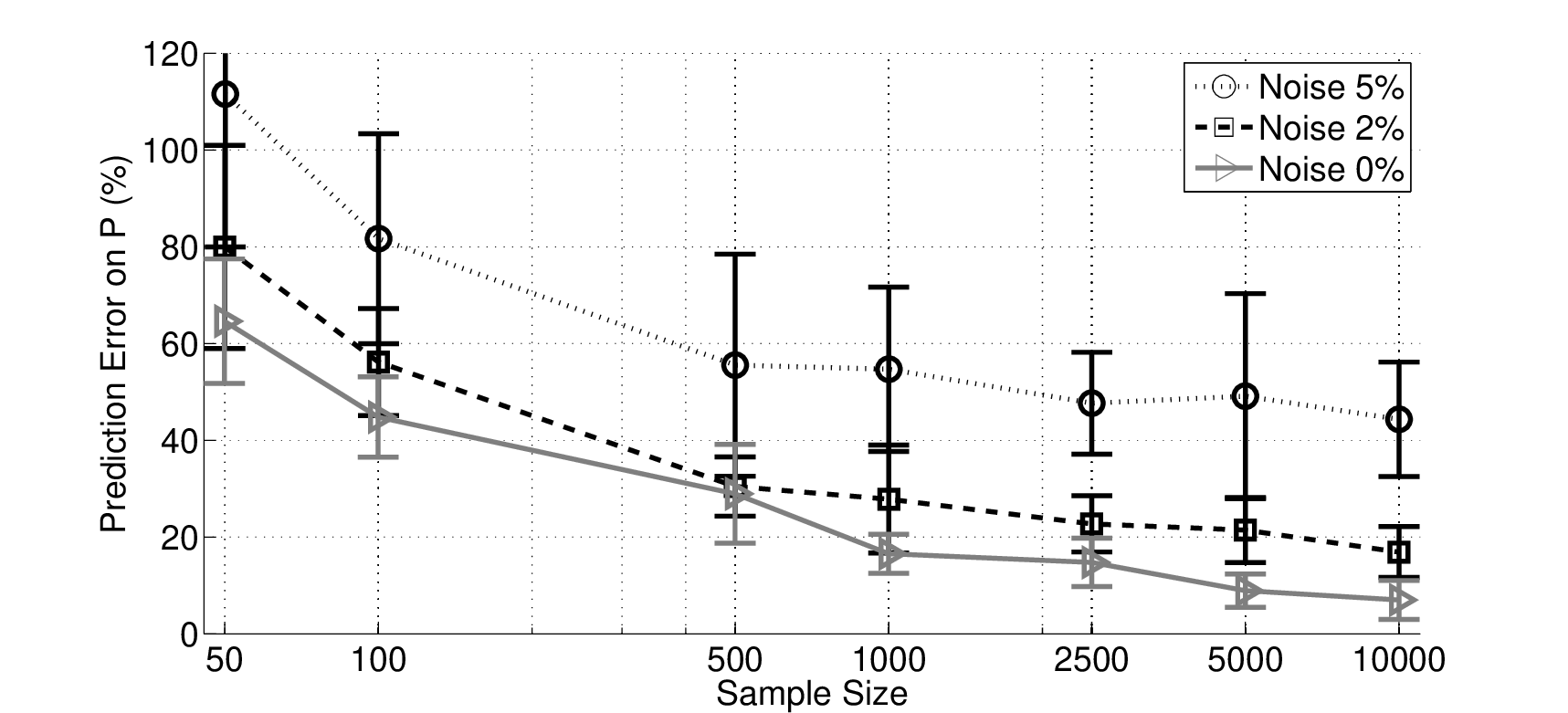}
\hspace{-0.3in} & \hspace{-0.3in} \includegraphics[width=2.7in]{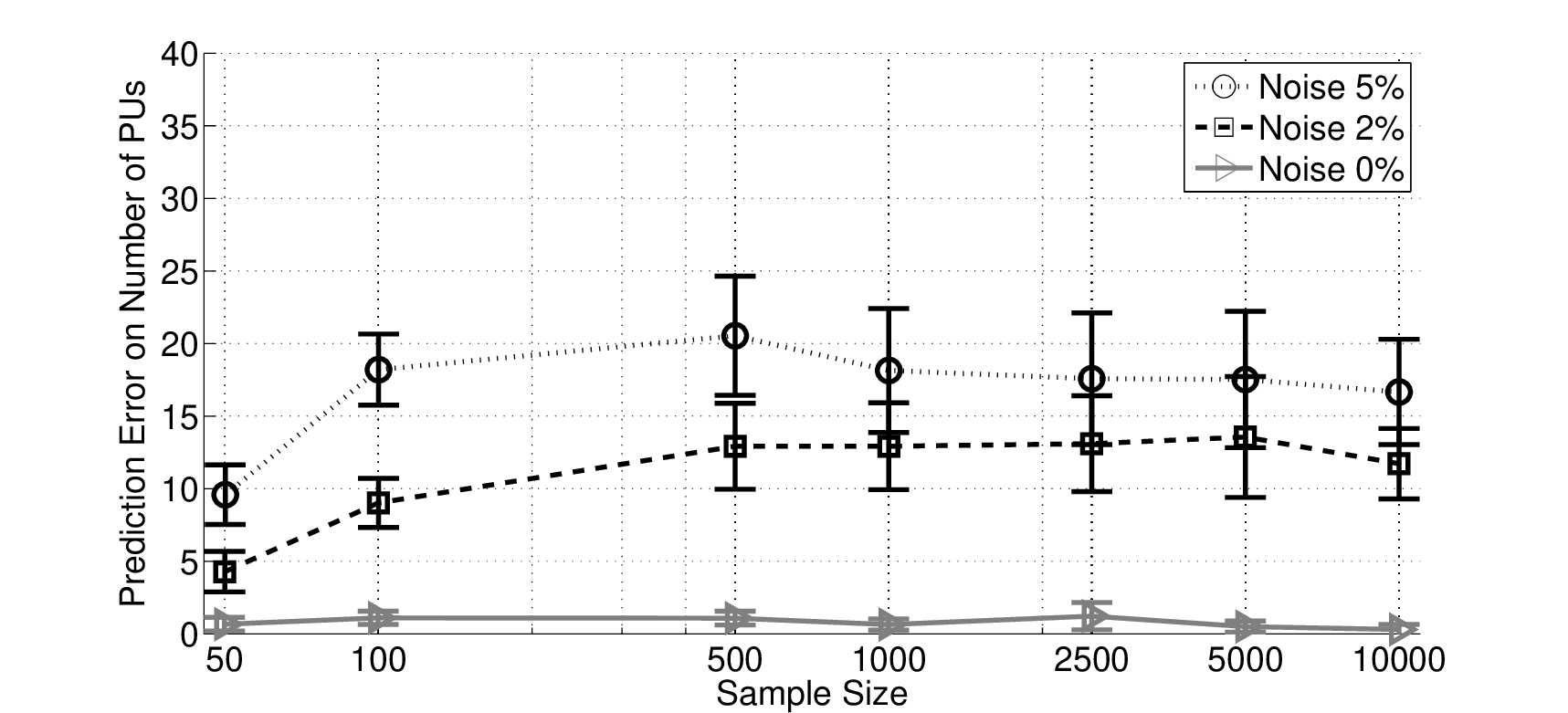}\\
\hspace{-0.15in}(a) Structure error ratio
\hspace{-0.3in} & \hspace{-0.3in} (b) PU active probability error ratio
\hspace{-0.3in} & \hspace{-0.3in} (b) Miscount of PUs \\
\end{tabular}
\begin{tabular}{cc}
\includegraphics[width=2.7in]{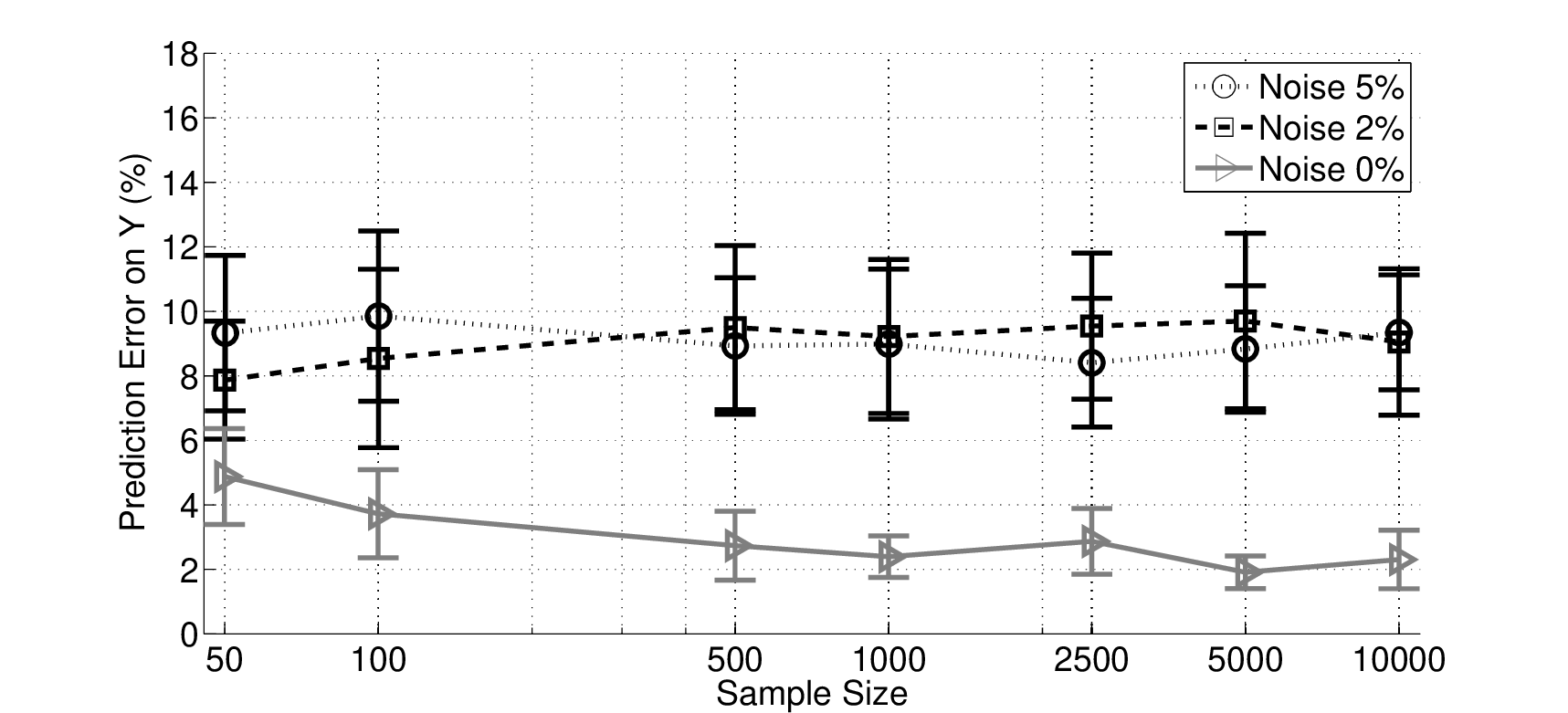}
\hspace{-0.3in} & \hspace{-0.3in} \includegraphics[width=2.7in]{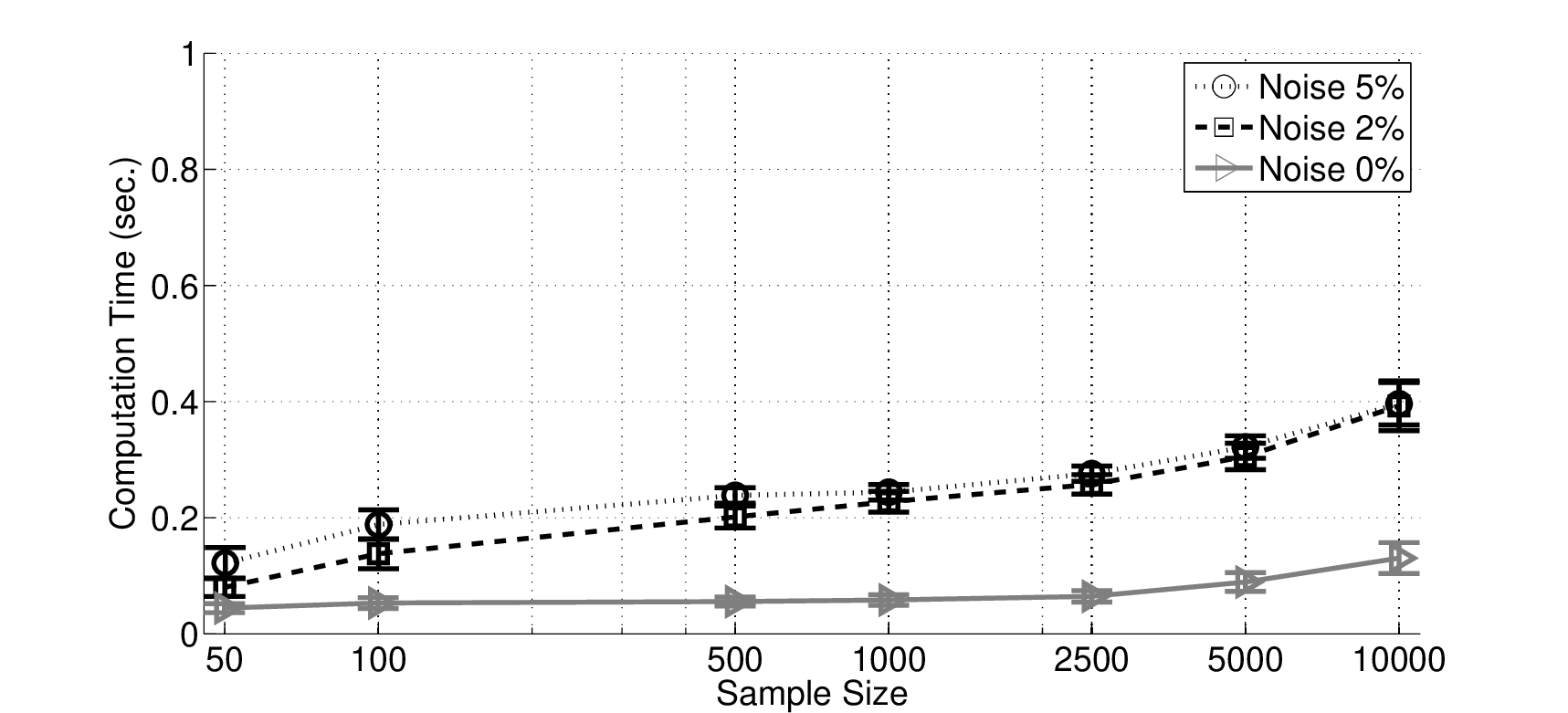}\\
(d) PU activity error ratio
\hspace{-0.3in} & \hspace{-0.3in} (e) Computation time\\
\end{tabular}
\end{center}
\caption{Effects of noise level and the size of observations on inference results.
The $x$-axis is in logarithmic scale.}
\label{fig:error2}
\end{figure*}


\subsubsection{Varying the size of observations $T$}
In the second set of experiments, we fix the number of PUs $n = 10$ and
study the impact of the observation size $T$. A small $T$ (and thus insufficient
observations) would lead to higher uncertainty while a large $T$ incurs higher
computation overhead. Furthermore, if $T$ is too small, some PUs may never be
active in the trace, making them impossible to be inferred. It is therefore
interesting to investigate the effect of $T$ on the accuracy and computation
overhead of the proposed algorithm.

Experiment results over 50 runs for each observation size are shown in
Figure~\ref{fig:error2}.  As expected, the structure error ratio and the PU
active probability error ratio reduce significantly as $T$ increases from 50 to
1,000 (Figure~\ref{fig:error2}(a) and (b)).  If we further increase $T$ from
1,000 to 10,000, the performance gain is somewhat marginal. However, the
computation time grows considerably since it takes longer to process the
observations (Figure~\ref{fig:error2}(e)).  From Figure~\ref{fig:error2}(c) and
(d), we also see that the miscount of PUs  and the PU activity error ratio are
not sensitive to the sample size $T$, but are more affected by the noise level.

\section{Related work}
\label{sec:related}
Independent component analysis (ICA) has been studied in the past as
a computational method for separating a multivariate signal into additive
subcomponents supposing the mutual statistical independence of the non-Gaussian
source signals. Most ICA methods
assume linear mixing of continuous signals~\cite{Hyvarinen00}. A special
variant of ICA, called Boolean Independent Component Analysis (BICA), considers boolean mixing (e.g., OR,
XOR etc.) of binary signals. Existing solutions to BICA mainly differ in their
assumptions of prior distribution of the mixing matrix, noise model, and/or
hidden causes. In \cite{Yeredor07}, Yeredor considers BICA in XOR mixtures and
investigates the identifiability problem. A deflation algorithm is proposed for
source separation based on entropy minimization. In \cite{Yeredor07} the number of
independent random sources $K$ is assumed to be known. Furthermore, the mixing
matrix is an $K$-by-$K$ invertible matrix.  In
\cite{computer06anon-parametric}, infinite number of hidden causes following
the same Bernoulli distribution are assumed.  Reversible jump Markov chain
Monte Carlo and  Gibbs sampler techniques are applied. In contrast, in our
model, the hidden causes may follow different distribution and the mixing
matrix tends to be sparse. Streich {\it et al.}~\cite{Streich09} study the problem of multi-assignment
clustering for boolean data, where the observations either from a signal
following OR mixtures or from a noise component. The key assumption made in this
work is that the elements of matrix $\X$ are conditionally independent given
the model parameters. This greatly reduces the computational complexity and
makes the scheme amenable to gradient descent optimization solution. This
assumption is in general invalid. In \cite{Kabán_factorisationand}, the problem of factorization and
de-noise of binary data due to independent continuous sources is considered,
which follow beta distribution. Finally, \cite{computer06anon-parametric}
consider under-presented case of less sensors than sources with continuous
noise, while \cite{Kabán_factorisationand} and~\cite{Streich09} deal with over-determined case, where the number of
sensors is much larger. In this work, we consider primarily the under-presented cases
encountered in data networks.

There exists a large body of work on blind deconvolution with binary sources in
the context of wireless communication~\cite{Diamantaras06,LiCichocki03}. In
time-invariant linear channels, the output signal $x(k)$ is a convolution of
the channel realizations $a(k)$ and the input signal $s(k)$, $k=1,2,\ldots, K$
as follows:
\beq
x(k) =
\sum_{l=0}^{L}{a(l)s(k-l)}, k=1,\ldots, K.
\label{eq:deconvolution}
\eeq
The objective is to recover the input signal $s$.  Both stochastic and
deterministic approaches have been devised for blind deconvolution.  As evident
from \eqref{eq:deconvolution}, the output signals are linear mixtures of the
input sources in time, and additionally the mixture model follows a specific
structure.

Literature on boolean/binary factor analysis (BFA) is also related to our work.
The goal of BFA is to decompose a binary matrix $\X_{m\times T}$ into $\A_{m
\times n} \otimes \B_{n\times T}$ with $\otimes$ being the \textit{OR} mixture
relationship as defined in (\ref{eq:boolean}). We use the same notation of $m$, $n$,
and $T$  to illustrate the relationship between BFA and \bica.  $\X$ in BFA is
often called an attribute-object matrix providing $m$-dimension attributes of
$T$ objects.  $\A$ and $\B$ are  the attribute-factor and factor-object
matrices. All the elements in $\X$, $\A$, and $\B$ are either 0 or 1. $n$ is
defined to be the number of underlying factors and is assumed to be
considerably smaller than the number of objects $T$. BFA methods aim to find a
feasible decomposition minimizing $n$.  Frolov {\it et al.} study the problem
of factoring a binary matrix in a series of papers~\cite{Frolov05, Frolov07,
Frolov07_2} using Hopfield neural networks. This approaches are based on
heuristics and do not provide much  theoretical insight regarding the
properties of the resulting decomposition.  More recently, Belohlavek {\it et
al.} propose a matrix decomposition method utilizing formal concept
analysis~\cite{Belohlavek2010}.  The paper claims that optimal decomposition
with the minimum number of factors are those where factors are formal concepts.
It is important to note that even though BFA assumes a similar disjunctive
mixture model to our problem, the objective is different.  While BFA tries to
find a matrix factorization so that the number of factors are minimized, \bica~tries
to identify independent components. One can easily come up an example,
where the number of independent components (factors) is larger than the number
of attributes, while BFA always finds factors no larger than the number of
attributes.
\section{Conclusions}
\label{sect:conclusion}
In this paper,  we introduced the PU separation problem for cognitive radio
networks and argue its relevance in collaborative spectrum sensing and monitor
resource allocation. We demonstrated that  a binary mixing model is sufficient
to characterize the behavior of energy detectors in presence of multiple PUs,
and devised a binary inference framework to resolve the PU separation problem.
The results are somewhat surprising that PUs can be accurately separated and
identified with only binary observations from the set of monitors to which they
are observable. Simulation validation shows that the PU-SU relationship as well
as the PUs' statistics and activities can be estimated with high accuracy when
the noise is marginal.
\section*{Acknowledge}
This work is funded in part by the National Science Foundation under grants
CNS-0953377, CNS-0905556, CNS-091046, CNS-0832084, and ECCS-1028782.
\bibliographystyle{IEEEtran}

\end{document}